\documentclass[11pt,a4]{article}

\usepackage{amssymb}
\usepackage{amsmath}

\usepackage[latin1]{inputenc}
\usepackage[linesnumbered,ruled]{algorithm2e}

\usepackage{graphicx}

\usepackage{tikz}
\usetikzlibrary{arrows,shapes.gates.logic.US,shapes.gates.logic.IEC,calc}

\usepackage{xy}
\xyoption{import}

\usepackage{srcltx}

\newcommand{\ds}                   {\displaystyle}

\newcommand{\mis}                 {{\text{ is }}}

\newcommand{\qurtains}             {\hfill{QED}}

\ifx \theorem \undefined
\newtheorem{definition}{\vspace{1mm}Definition}[section]

\newtheorem{pth}[definition]{\vspace{1mm}Theorem}
\newtheorem{corollary}[definition]{\vspace{1mm}Corollary}

\newenvironment{proof}{\begin{trivlist}\item{\bf Proof:}}{\qurtains\end{trivlist}}

\newcommand{\nats}                  {{\mathbb N}}
\newcommand{\rats}                  {{\mathbb Q}}

\newcommand{\reals}                 {{\mathbb R}}

\newcommand{\termvalue}[1]          {{\lbrack\!\lbrack #1 \rbrack\!\rbrack}}
\newcommand{\tv}[1]                 {{\termvalue{#1}}}
\newcommand{\norm}[1]               {{\|#1\|}}

\newcommand{\ovl}[1]                {{\overline{#1}}}

\newcommand{\from}                  {{{\leftarrow}}}

\newcommand{\lverum}                {{\mathsf{t\!t}}}
\newcommand{\lfalsum}               {{\mathsf{f\!f}}}

\newcommand{\lneg}                  {\mathop{\neg}}

\newcommand{\limp}                  {\mathbin{\supset}}
\newcommand{\leqv}                  {\mathbin{\equiv}}
\newcommand{\lconj}                 {\mathbin{\wedge}}
\newcommand{\ldisj}                 {\mathbin{\vee}}

\newcommand{\lmaj}                  {\textsf{M}}

\newcommand{\ovn}[1]                {{\ovl{#1}}}

\newcommand{\lnneg}                 {\mathop{\ovn{\lneg}}}
\newcommand{\lnimp}                 {\mathbin{\ovn{\limp}}}
\newcommand{\lneqv}                 {\mathbin{\ovn{\leqv}}}
\newcommand{\lnconj}                {\mathbin{\ovn{\lconj}}}
\newcommand{\lndisj}                {\mathbin{\ovn{\ldisj}}}

\newcommand{\lnmaj}                 {\ovn{\textsf{M}}}

\newcommand{\ovf}[1]                {{\widetilde{#1}}}

\newcommand{\lfneg}                 {\mathop{\ovf{\lneg}}}
\newcommand{\lfimp}                 {\mathbin{\ovf{\limp}}}
\newcommand{\lfeqv}                 {\mathbin{\ovf{\leqv}}}
\newcommand{\lfconj}                {\mathbin{\ovf{\lconj}}}
\newcommand{\lfdisj}                {\mathbin{\ovf{\ldisj}}}

\newcommand{\lfnneg}                {\mathop{\ovf{\lnneg}}}
\newcommand{\lfnimp}                {\mathbin{\ovf{\lnimp}}}
\newcommand{\lfneqv}                {\mathbin{\ovf{\lneqv}}}
\newcommand{\lfnconj}               {\mathbin{\ovf{\lnconj}}}
\newcommand{\lfndisj}               {\mathbin{\ovf{\lndisj}}}

\newcommand{\lfmaj}                 {\ovf{\lmaj}}
\newcommand{\lfnmaj}                {\ovf{\lnmaj}}

\newcommand {\potc}[1]             {\mathbin{\,\sqsubseteq_{#1}\,}}
\newcommand {\otc}                  {\mathbin{\;\sqsubseteq\;}}

\newcommand {\Pnu}  		{{\mathfrak{P}}}

\newcommand {\num}[1]              {{#1}}

\newcommand{\sat}                    {\Vdash}
\newcommand{\ent}                    {\vDash}

\newcommand{\satfo}                  {\mathbin{\,\sat^{\mathsf{fo}}\,}}

\newcommand{\satuc}                  {\mathbin{\,\sat^{\mathsf{uc}}\,}}
\newcommand{\nsatuc}                  {\mathbin{\,\not\sat^{\mathsf{uc}}\,}}
\newcommand{\entuc}                  {\mathbin{\,\ent^{\mathsf{uc}}\,}}
\newcommand{\nentuc}                  {\mathbin{\,\not \ent^{\mathsf{uc}}\,}}

\newcommand {\Sigmaf}               {\ovf\Sigma}

\newcommand {\Lc}                   {L^{\mathsf{c}}}
\newcommand {\Lo}                   {L^{\mathsf{o}}}
\newcommand {\La}                   {L^{\mathsf{a}}}
\newcommand {\Lfc}                   {L^{\mathsf{uc}}}

\newcommand {\ORCF}                   {\mathsf{RCOF}}

\newcommand {\UCL}                  {\mathsf{UCL}}

\newcommand {\var}                   {\text{var}}
\newcommand {\fau}                   {\text{unr}}
\newcommand {\SATV}               {\text{MC}}
\newcommand {\SAT}               {\text{SAT}}

\newcommand {\VAL}               {\text{VAL}}

\newcommand {\ARR}               {\text{RRA}}

\newcommand {\RRD}               {\text{RRD}}

\newcommand {\PMC}               {\text{MW}}
\newcommand {\SATU}               {\text{SAT}}
\newcommand {\ENTA}               {\text{acVAL}}
\newcommand {\OSC}               {\text{SRO}}
\newcommand {\CQE}               {\text{CQE}}

\newcommand {\orcf}                  {\text{rcof}}
\newcommand {\cA}                   {\mathcal{A}}

\newcommand{\nsatfo}                  {\mathbin{\,\not\sat^{\mathsf{fo}}\,}}

\newcommand {\PL}                  {\mathsf{PL}}

\newcommand {\Sigmafc}              {\Sigma^{\mathsf{uc}}}
\newcommand{\satfc}                  {\mathbin{\,\sat^{\mathsf{uc}}\,}}
\newcommand{\entfc}                  {\mathbin{\,\ent^{\mathsf{uc}}\,}}

\begin{document}

\title{Decision and optimization problems in the Unreliable-Circuit Logic}

\author{J.~Rasga${}^1$ \ \ C.~Sernadas${}^1$ \ \ P.~Mateus${}^2$ \ \ A.~Sernadas${}^1$ \\[1mm]
{\scriptsize ${}^1$ Dep.~Matemática, Instituto Superior Técnico}\\[-1.5mm] 
{\scriptsize and}\\[-1.5mm] 
{\scriptsize Centro de Matemática, Aplicações Fundamentais e Investigação Operacional}\\[0.5mm]
{\scriptsize ${}^2$ Dep.~Matemática, Instituto Superior Técnico}\\[-1mm]
{\scriptsize and}\\[-1.5mm]
{\scriptsize Instituto de Telecomunicações}\\[3mm]
{Universidade de Lisboa, Portugal}\\[3mm]
{\tiny \{joao.rasga,cristina.sernadas,paulo.mateus,amilcar.sernadas\}@tecnico.ulisboa.pt}}

\date{July 27, 2016}

\maketitle

\begin{abstract}
The ambition constrained validity  and the model witness problems  in the logic $\UCL$, proposed in \cite{acs:jfr:css:pmat:13},  for reasoning about circuits with unreliable gates are analyzed.  Moreover, two additional problems, motivated by the applications,  are studied. One consists of finding bounds on the reliability rate of the gates that ensure that a given circuit has an intended success rate. The other consists of
finding a reliability rate of the gates that maximizes the success rate of a given circuit. Sound and complete algorithms are developed for these problems and their computational complexity is studied. \\[2mm]
\end{abstract}

\section{Introduction}

Classical propositional logic is the right setting for the design and verification of logic circuits represented by formulas. Some examples of recent work in this broad area can be found in~\cite{cla:09,bol:10}. 

In practice, 
logic circuits can be built with unreliable gates that can produce the wrong output by fortuitous misfiring, hopefully with a very low probability.
The interest on reasoning about such unreliable circuits was pointed out by John von Neumann~\cite{neu:56} and 
recently reawaken by developments towards nano-circuits, see for instance~\cite{han:11,lee:12,rej:09}. 
In fact, in nano circuits, the extremely low level of energy carried by each gate leads to a higher probability of it being disturbed by the environment and, so, misfiring. This was the motivation for proposing, see~\cite{acs:jfr:css:pmat:13}, the logic $\UCL$ for reasoning about unreliable circuits, as an extension of propositional logic (see also~\cite{acs:css:jfr:pmat:11,jfr:css:acs:14} for more work on this topic). 

In this paper, we investigate several computational and decision problems related with $\UCL$, namely the ambition constrained validity problem,
the model witness problem, the reliability rate abduction problem and the success rate optimization problem. The first two are common when analyzing a logic from an algorithmic and complexity  points of view. The remaining constitute applications of $\UCL$ to real problems.
For all the problems we discuss an application scenario.

The ambition constrained validity problem is a decision problem for verifying whether or not a formula representing a circuit with unreliable gates is a semantic consequence of a finite set of ambition formulas. An algorithm for this problem is presented and shown to be in PSPACE and a restricted version of it is proved to be co-NP complete (for computational complexity issues the reader may consult~\cite{hom:11}). 

The model witness problem is a constructive model checking problem where the objective is to give a model for a formula representing a circuit with unreliable gates whenever there is one. An algorithm for this problem is presented and shown to run deterministically in polynomial space. Moreover, we also show that a logarithmic version of the satisfiability problem for $\UCL$  is NP complete.

The reliability rate abduction problem is a computational problem that given a formula representing an unreliable circuit, an intended success rate of the circuit and a natural number returns a set of intervals with limits determined by the given natural number such that
if the reliability rate of the gates falls in those intervals  then the circuit has the given success rate. Herein we define an algorithm for this problem that runs deterministically in polynomial space. Moreover, we show that a restricted version of the corresponding decision problem is in P.

Finally, we investigate the success rate optimization problem which is a computational problem that given a formula representing an unreliable circuit returns the maximum success rate of the circuit as well as a reliability rate of the gates that ensures that success rate. We provide an algorithm for this problem and show that it runs deterministically in exponential time.

For solving these problems we make extensive use of the decidable theory $\ORCF$ of real closed ordered fields (see~\cite{tar:51}) namely some algorithms for this theory (see~\cite{bas:06}).

We provide an overview of $\UCL$ in Section~\ref{sec:overview}. Sections~\ref{sec:gval},~\ref{sec:mcp},~\ref{sec:rra} and~\ref{sec:sro} are dedicated to the computation and decision problems described above. Finally, in Section~\ref{sec:outlook}, we point out further developments.

\section{An overview of $\UCL$}\label{sec:overview}

The unreliable-circuit logic $\UCL$ is an extension of propositional logic and was introduced in \cite{acs:jfr:css:pmat:13} for reasoning about logic circuits with single-fan-out unreliable gates that can produce the wrong output by fortuitous misfiring. Herein, we provide an overview of its syntax and semantics. 

We start by presenting a modicum of  the theory of real closed ordered fields, $\ORCF$, and of propositional logic, $\PL$, that we need later on for presenting $\UCL$. 


The first-order signature $\Sigma_\orcf$ of $\ORCF$ contains the constants $\num0$ and $\num1$, the unary function symbol $-$, the binary function symbols $+$ and $\times$, and the binary predicate symbols $=$ and $<$. 
As usual, we may write $t_1 \leq t_2$ for $(t_1 < t_2) \ldisj (t_1=t_2)$ and 
$t_1\,t_2$ for $t_1 \times t_2$. In the sequel, we denote by $\termvalue{t}^\rho$ the denotation of term $t$ over the $\ORCF$ structure based on  $\reals$ and the assignment $\rho$.
We use $\satfo$ for denoting   satisfaction in first-order logic. In the sequel, we use extensively the fact that  the theory $\ORCF$ is decidable~\cite{tar:51}.


We need an enriched signature of propositional logic that we denote by $\Sigma$. 
Let $\Sigma$ be the signature for $\PL$ containing 
the propositional constants $\lverum$ (verum) and $\lfalsum$ (falsum) plus the propositional connectives $\lneg$ (negation), $\lconj$ (conjunction), $\ldisj$ (disjunction), $\limp$ (implication), $\leqv$ (equivalence) and $\lmaj_{3+2k}$ ($k$-ary majority) for each $k\in\nats$, as well as their negated-output counterparts $\lnneg$ (identity), $\lnconj$ (negated conjunction), $\lndisj$ (negated disjunction), 
$\lnimp$ (negated implication), $\lneqv$ (negated equivalence) and $\lnmaj_{3+2k}$ ($k$-ary negated majority) for each $k\in\nats$. 
We denote by $L(X)$ the set of propositional formulas over $\Sigma$ and  a set $X$ of propositional variables.  
Given a formula $\varphi \in L(X)$ and a valuation $v:X\to\{\bot,\top\}$, we write 
$v \sat \varphi$, 
for saying that $v$ satisfies formula $\varphi$.

We now are ready to review the unreliable-circuit logic.~The {\it signature} of $\UCL$ is the triple $(\Sigmafc,\nu,\mu)$ where:
\begin{itemize}
\item $\Sigmafc$ contains $\Sigma$ and the following additional connectives used for representing the unreliable gates: 
$$\text{$\lfneg$, $\lfnneg$, $\lfconj$, $\lfnconj$, 
$\lfdisj$, $\lfndisj$, $\lfimp$, $\lfnimp$, $\lfeqv$, $\lfneqv$, $\lfmaj_{3+2k}$ and $\lfnmaj_{3+2k}$;}$$
\item both $\nu$ and $\mu$ are symbols used for denoting probabilities.
\end{itemize}

Each unreliable gate is assumed to produce the correct output with probability $\nu$. 
A circuit is accepted as good if it produces the correct output with probability not less than $\mu$.

We denote by $\Sigmaf$ the subsignature of the unreliable connectives in $\Sigmafc$. Thus,
$\Sigmafc = \Sigma \cup \Sigmaf$. 
Moreover, for each $n\in\nats$, we denote by 
$\Sigma_n$, $\Sigmafc_n$ and $\Sigmaf_n$ the set of $n$-ary constructors 
in $\Sigma$, $\Sigmafc$ and $\Sigmaf$, respectively.  
Plainly, $\Sigmaf_0=\emptyset$.
Given a $\PL$ formula $\varphi$ 
and a $\UCL$ formula $\psi$, we write
$$\varphi \otc \psi$$
for saying that $\varphi$ is a possible {\it outcome} of $\psi$. This outcome relation is inductively defined as expected:
\begin{itemize}
\item $\varphi \otc \varphi$ provided that $\varphi$ is a $\PL$ formula;
\item $c(\varphi_1,\dots,\varphi_n) \otc c(\psi_1,\dots,\psi_n)$ 
		provided that
		$n\geq1$, 
		$c\in\Sigma_n$ 
		and $\varphi_i \otc \psi_i$ for $i=1,...,n$;
\item $c(\varphi_1,\dots,\varphi_n) \otc \ovf{c}(\psi_1,\dots,\psi_n)$ 
		provided that
		$\ovf{c}\in\Sigmaf_n$,
		and $\varphi_i \otc \psi_i$ for $i=1,...,n$;
\item $\ovn{c}(\varphi_1,\dots,\varphi_n) \otc \ovf{c}(\psi_1,\dots,\psi_n)$ 
		provided that
		$\ovf{c}\in\Sigmaf_n$,
		and $\varphi_i \otc \psi_i$ for $i=1,...,n$.
\end{itemize}
For each such $\psi$, we denote by 
$$\Omega_\psi$$
the set $\{\varphi:\varphi\otc\psi\}$ of all possible outcomes of $\psi$.
Clearly, $\Omega_\varphi=\{\varphi\}$ for each $\PL$ formula $\varphi$.

In $\UCL$, by a {\it term} we mean a univariate polynomial written according to the term syntax of $\ORCF$, using $\nu$ as the variable. 
Symbol $\mu$ is also taken as a variable in the context of $\ORCF$ but it is not used in $\UCL$ terms.

Three kinds of formulas are needed for reasoning about circuits with unreliable gates: 
\begin{itemize}

\item {\it Circuit formulas} or {\it c-formulas} that are propositional formulas built with the symbols in $\Sigmafc$ and $X$.
These c-formulas can be used for representing unreliable circuits.
For instance, the c-formula
$$(x_1 \lfdisj (\lfneg x_2)) \lfconj x_3$$
represents the unreliable circuit in Figure~\ref{fig:circuit1}.
Circuit formulas can also be used for asserting relevant properties of unreliable circuits.
For example, 
given the c-formula $\psi$ and the $\PL$ formula $\varphi$,
the c-formula
$$\psi \leqv \varphi$$
is intended to state that the unreliable circuit represented by $\psi$ can be accepted as equivalent to
the reliable circuit represented by $\varphi$, in the sense that the two circuits agree with probability of at least $\mu$.

\begin{figure}
\centering
\tikzstyle{branch}=[fill,shape=circle,minimum size=3pt,inner sep=0pt]
\begin{tikzpicture}[label distance=2mm]

    \node (x1) at (0,0) {$x_1$};
    \node (x2) at (0.6,0) {$x_2$};
    \node (x3) at (1.2,0) {$x_3$};

    \node[not gate US, draw] at ($(x3)+(1.4,-3)$) (Not) {$\sim$};
    \node[or gate US, draw, anchor=input 1] at ($(Not.output)+(1,1)$) (Or) {${}\sim{}$};
    \node[and gate US, draw, anchor=input 1] at ($(Or.output)+(1,-1.5)$) (And) {${}\sim{}$};

    \path let \p1=(x1), \p2=(Or.input 1) in node[branch] (x1oru) at ($(\x1,\y2)$)   {};
    \path let \p1=(x1), \p2=(And) in node[] (x1andm) at ($(\x1,\y2)+(0.0,-1.0)$)   {};

    \path let \p1=(x2), \p2=(Not.input) in node[branch] (x2not) at ($(\x1,\y2)$)   {};
    \path let \p1=(x2), \p2=(And) in node[] (x2andm) at ($(\x1,\y2)+(0.0,-1.0)$)   {};

    \path let \p1=(x3), \p2=(And.input 2) in node[branch] (x3andp) at ($(\x1,\y2)$)   {};
    \path let \p1=(x3), \p2=(And) in node[] (x3andm) at ($(\x1,\y2)+(0.0,-1.0)$)   {};

    \draw (x1) -- (x1oru);
    \draw (x1oru) -- (Or.input 1);
    \draw (x1oru) -- (x1andm);

    \draw (x2) -- (x2not);
    \draw (x2not) -- (Not.input);
    \draw (x2not) -- (x2andm);

    \draw (x3) -- (x3andp);
    \draw (x3andp) |- (And.input 2);
    \draw (x3andp) -- (x3andm);

    \draw (Not.output) -- ([xshift=0.4cm]Not.output) |- (Or.input 2);
    \draw (Or.output) -- ([xshift=0.4cm]Or.output) |- (And.input 1);
    \draw (And.output) -- ([xshift=0.5cm]And.output) node[above] {};

\end{tikzpicture}\vspace*{-6mm}
\caption{Circuit represented by the c-formula $(x_1 \lfdisj (\lfneg x_2)) \lfconj x_3$.}
\label{fig:circuit1}
\end{figure}
		
\item {\it Outcome formulas} or {\it o-formulas} that are of the general form
$$\Phi \potc{P} \psi$$
where $\psi$ is a c-formula, 
$\Phi\subseteq\Omega_\psi$ and $P$ is a term. 
Such an o-formula is used with the intent of stating that 
the probability of the outcome of $\psi$ being in $\Phi$ is at least $P$.
For instance,
$$\{(x_1 \ldisj (\lneg x_2)) \lconj x_3, (x_1 \lndisj (\lneg x_2)) \lconj x_3\} \potc{\nu^2} (x_1 \lfdisj (\lfneg x_2)) \lfconj x_3$$
should be true in any interpretation of $\UCL$ because
$(x_1 \ldisj (\lneg x_2)) \lconj x_3$ and $(x_1 \lndisj (\lneg x_2)) \lconj x_3$
are both possible outcomes of $(x_1 \lfdisj (\lfneg x_2)) \lfconj x_3$
(the former when all the unreliable gates perform perfectly and the latter when only the OR gate fails),
the probability of the former is $\nu^3$,
the probability of the latter is $(\num1-\nu)\nu^2$,
and $\nu^3 + (\num1-\nu)\nu^2 = \nu^2$.
We may use $\varphi_1,\dots,\varphi_m \potc{P} \psi$ instead of $\{\varphi_1,\dots,\varphi_m\} \potc{P} \psi$.
		
\item {\it Ambition formulas} or {\it a-formulas} that are of the general form
$$\mu \leq P$$
where $P$ is a term. 
Such an a-formula can be used for constraining
the envisaged non-failure probability $\mu$ of the overall circuit.
For example, every interpretation of $\UCL$ where the a-formula
$$\mu \leq \nu^2 + (\num1-\nu)^2$$
holds should make 
$$(\lneg(x_1 \ldisj x_2)) \leqv (\lfneg(x_1 \lfdisj x_2))$$
true, since
$\lneg(x_1 \ldisj x_2)$
and
$\lnneg(x_1 \lndisj x_2)$
are the outcomes of the circuit at hand
$\lfneg(x_1 \lfdisj x_2)$
that make it in agreement to the ideal one $\lneg(x_1 \ldisj x_2)$,
the probability of outcome $\lneg(x_1 \ldisj x_2)$ is $\nu^2$,
the probability of outcome $\lnneg(x_1 \lndisj x_2)$ is $(\num1-\nu)^2$,
and, so, their aggregated probability is 
$$\nu^2 + (\num1-\nu)^2.$$
\end{itemize}

We denote by $\Lc(X)$, $\Lo(X)$ and $\La$ the set of c-formulas, o-formulas and a-formulas, respectively, 
and by $\Lfc(X)$ the set $\Lc(X) \cup \Lo(X) \cup \La$ of all $\UCL$ formulas. Observe that each of these sets is decidable.
Given a c-formula $\psi$ and $\varphi\in\Omega_\psi$, we write
$$\Pnu[\psi \triangleright \varphi]$$
for the $\UCL$ term that provides the probability of outcome $\varphi$ of $\psi$. This term is inductively defined as follows:
\begin{itemize}
\item $\Pnu[\varphi \triangleright \varphi] \mis \num1$ for each $\varphi\in L(X)$;
\item $\Pnu[c(\psi_1,\dots,\psi_n) \triangleright c(\varphi_1,\dots,\varphi_n)]\mis 
		\ds \prod_{i=1}^n \Pnu[\psi_i \triangleright \varphi_i]$	
		 for each
		$n\geq1$, $c\in\Sigma_n$ and $\varphi_i\otc\psi_i$ for $i=1,\dots,n$;
\item $\Pnu[\ovf{c}(\psi_1,\dots,\psi_n) \triangleright c(\varphi_1,\dots,\varphi_n)]\mis 
		\nu \ds\prod_{i=1}^n \Pnu[\psi_i \triangleright \varphi_i]$		
		for each
		$\ovf{c}\in\Sigmaf_n$ and $\varphi_i\otc\psi_i$ for $i=1,\dots,n$;
\item $\Pnu[\ovf{c}(\psi_1,\dots,\psi_n) \triangleright \ovn{c}(\varphi_1,\dots,\varphi_n)]\mis 
		(\num1- \nu) \ds \prod_{i=1}^n \Pnu[\psi_i \triangleright \varphi_i]$
		for each
		$\ovf{c}\in\Sigmaf_n$ and $\varphi_i\otc\psi_i$ for $i=1,\dots,n$.
\end{itemize}

For instance,
$$\Pnu[\lfneg(x_1 \lfdisj x_2) \triangleright \lneg(x_1 \lndisj x_2)]$$
is the polynomial
$$\nu(\num1-\nu)$$
since, for the given input provided by $x_1$ and $x_2$,
outcome $\lneg(x_1 \lndisj x_2)$ happens when $\lfneg$ behaves as it should and $\lfdisj$ fails,
that is, when $\lfneg$ produces the correct output and $\lfdisj$ misfires.


Each interpretation of $\UCL$ should provide
a valuation to the variables in $X$,
a model of $\ORCF$ and 
an assignment to the variables $\nu$ and $\mu$. 
However,
the choice of the model of $\ORCF$ is immaterial since all such models are elementarily equivalent (see Corollary~3.3.16 in~\cite{mar:02}) and, so, we adopt the ordered field $\reals$ of the real numbers.
Thus, by an interpretation of $\UCL$ we mean a pair
$$I=(v,\rho)$$
where $v$ is a propositional valuation and $\rho$ is an assignment over $\reals$ such that:
$$\begin{cases}
{\frac{1}{2}} < \rho(\mu) \leq 1\\
{\frac{1}{2}} < \rho(\nu) \leq 1.
\end{cases}$$

We now proceed to define {\it satisfaction}, by an interpretation $I=(v,\rho)$. 
Starting with c-formulas, we write
$$I \satfc \psi$$
for stating that
$$\reals\, \rho \satfo \ds\mu\leq\sum_{\text{\shortstack[c]
{$\varphi \otc \psi$\\
$v \sat \varphi$}
}} \Pnu[\psi \triangleright \varphi].$$
That is, the aggregated probability of the outcomes of $\psi$ that are (classically) satisfied by $v$ is at least the value of $\mu$.
Concerning o-formulas, we write
$$I \satfc \Phi \potc{P} \psi$$
for stating that
$$\reals\, \rho \satfo {P} \leq \ds\sum_{\varphi\in\Phi} \Pnu[\psi \triangleright \varphi].$$
That is, the collection $\Phi$ of possible outcomes of $\psi$ has aggregated probability greater than or equal to the value of 
$P$.
Finally, concerning a-formulas, we write
$$I \satfc \mu \leq P$$
for stating that
$$\reals\, \rho \satfo \mu \leq {P}.$$
That is, the required probability $\rho(\mu)$ for the correct output being produced by the whole circuit does not exceed the value of $P$.
Satisfaction is extended to mixed sets of o-formulas, c-formulas and a-formulas with no surprises. 
Given $\Gamma\subseteq \Lfc(X)$,
$I \satfc \Gamma$
if $I \satfc \gamma$ for each $\gamma\in\Gamma$.
Then, entailment and validity in $\UCL$ are also defined as expected. Given 
$\{\theta\}\cup\Gamma \subseteq \Lfc(X)$, we write
$$\Gamma \entfc \theta$$
for stating that $\Gamma$ {\it entails} $\theta$ in the following sense:
$$I \satfc \theta \,\text{ whenever }\,I \satfc \Gamma, 
					\text{ for every interpretation }I.$$
Finally, we write 
$\entfc \theta$
for $\emptyset \entfc \theta$.

\section{Ambition constrained validity problem}\label{sec:gval}

We discuss the complexity of the ambition constrained version of the validity problem for $\UCL$ that consists on determining whether a formula representing a circuit with unreliable gates can be entailed by a finite set of ambition formulas. This problem coincides with the usual formulation of the validity problem for $c$-formulas whenever the set of ambition formulas is empty. It can be used, for example, to prove that a circuit with unreliable gates behaves as its ideal counterpart at least with a certain probability. 

The {\it ambition constrained validity problem} is the map 
$$\ENTA^\UCL: \Lc(X) \times \wp_{\text{fin}} \La\to \{0,1\}$$
that given a  formula $\psi$ representing a circuit with unreliable gates and a finite set $\Gamma$ of ambition formulas, returns $1$ if and only if $\Gamma \entuc \psi$.

In order to propose an algorithm for this problem, we need to refer to algorithms for the problems $\SATV^\PL$ and $\SAT^{\exists\ORCF}$.                                             
The $\SATV^\PL$ model checking  problem for propositional logic is the map 
$$\SATV^\PL: L(X) \times V \to \{0,1\}$$ 
that given a propositional formula $\varphi$ and a valuation $v$ returns $1$ if and  only if  $v \sat \varphi$.  Observe that
there is an algorithm (see~\cite{hom:11}) running in polynomial time for this problem. Herein, 
we use the name $\cA_{\SATV^\PL}$  to refer to such  algorithm.

To introduce the problem $\SAT^{\exists\ORCF}$, we need to refer to the  language $L^\exists(\Sigma_\orcf)$ consisting of the closed  formulas over 
the first-order signature $\Sigma_\orcf$  for the theory of real closed ordered fields (see~\cite{tar:51}) 
of the form $\exists x_1 \dots \exists x_k \, \theta$ where $\theta$ is a quantifier free formula. 
More precisely, the satisfiability problem for the existential fragment of the theory of real closed ordered fields is the map 
$$\SAT^{\exists\ORCF}: L^\exists(\Sigma_\orcf) \to \{0,1\}$$
that given a formula in $L^\exists(\Sigma_\orcf)$  returns $1$ if and only if there is an real closed ordered field that satisfies the formula. Observe that there is an algorithm for this problem that given an existential  formula $\eta$ with $s$ polynomials each of a degree at most $d$, over a set of $k$ variables, 
executes at most $s^{k+1} d^{O(k)}$ operations (see Theorem~13.13 in~\cite{bas:06}) in the ring generated by the coefficients of the polynomials. 
We use the name $\cA_{\SAT^{\exists\ORCF}}$ to refer to this algorithm. 

Moreover, we need also some notation. 
Given a c-formula $\psi$, we denote by $\var(\psi)$  the set of all propositional variables in $\psi$, by  $\fau(\psi)$ the sequence of 
all unreliable connectives occurring in $\psi$, and by $V_{\var(\psi)}$ the set of all valuations from $\var(\psi)$ to $\{\bot,\top\}$. Moreover, we use $|\cdot |$ for the map that returns the number of elements of the argument and 
$\norm{\cdot }$ for the map that returns the size in bits of the argument.

\begin{figure}
	       \hrule  \vspace*{4mm}
	       Inputs: c-formula $\psi$ and finite set $\Gamma$ of a-formulas.
\begin{enumerate}
  \item For each $v \in V_{\var(\psi)}$ do:
  \begin{enumerate}
    \item Let $P^v_{\psi} := 0$;
    \item For each $\varphi \in \Omega_\psi$:
    \begin{enumerate}
       \item If $\cA_{\SATV^\PL}(\varphi,v)=1$ then 
       $P^v_{\psi}:=P^v_{\psi}+ \Pnu[\psi \triangleright \varphi]$;
    \end{enumerate}
    \item If $\cA_{\SAT^{\exists\ORCF}}\left(\exists  \mu\exists \nu \lneg \left(\frac{ 1}{ 2} < \mu,\nu \leq  1 \limp \left(\bigwedge \Gamma \limp 
    P^v_{\psi} \geq \mu\right)\right)\right)=1$ then 
     Return $0$;
  \end{enumerate}
  \item Return $1$.
  \end{enumerate}
       \hrule  \vspace*{4mm}
\caption{Algorithm $\cA_{\ENTA^\UCL}$  for Problem $\ENTA^\UCL$.}
\label{alg:enta}
\end{figure}

\paragraph{Soundness and completeness}\ \\[1mm]
Let $\cA_{\ENTA^\UCL}$ be the algorithm in Figure~\ref{alg:enta}. Observe that it is in fact an algorithm since the cycles in Step 1 and Step 1(b) are over finite sets $V_{\var(\psi)}$ and $\Omega_\psi$, respectively, and 
$\cA_{\SATV^\PL}$ and $\cA_{\SAT^{\exists\ORCF}}$ are algorithms. 
Moreover, it is straightforward to conclude that the execution of the algorithm returns either $0$ or $1$.

\begin{pth} \em
Let $\psi \in \Lc(X)$ and $\Gamma \subseteq \La$. Then, 
$$\cA_{\ENTA^\UCL}(\psi,\Gamma)  \; \text{returns} \; 1 \quad \text{iff} \quad \ENTA^\UCL(\psi,\Gamma)=1.$$
\end{pth} 
\begin{proof}\ \\
$(\to)$ Assume that $\cA_{\ENTA^\UCL}(\psi,\Gamma)$ returns $1$. Then, for every valuation $v$ 
$$(\dag) \quad \cA_{\SAT^{\exists\ORCF}}\left(\exists  \mu\exists \nu \lneg \left(\frac{ 1}{ 2} < \mu,\nu \leq  1 \limp \left(\bigwedge \Gamma \limp 
    P^v_{\psi} \geq \mu\right)\right)\right)=0.$$
  Let $(v,\rho)$ be an interpretation of $\UCL$. Assume that $(v,\rho) \satuc \Gamma$. By $(\dag)$
  $$\reals\rho \satfo \frac{ 1}{ 2} < \mu,\nu \leq  1 \limp \left(\bigwedge \Gamma \limp 
    P^v_{\psi} \geq \mu\right).$$ Thus,
   $$\reals\rho \satfo P^v_{\psi} \geq \mu$$
   and so $(v,\rho) \satuc \psi$.
\\[2mm]
$(\from)$ We prove the result by contraposition. Assume that $\cA_{\ENTA^\UCL}(\psi,\Gamma)$ returns $0$. 
Thus, 
$\cA_{\ENTA^\UCL}(\psi,\Gamma)$ returns $0$ in step 1(c) for a valuation $v$. Then,
$$\cA_{\SAT^{\exists\ORCF}}\left(\exists  \mu\exists \nu \lneg \left(\frac{ 1}{ 2} < \mu,\nu \leq  1 \limp \left(\bigwedge \Gamma \limp 
    P^v_{\psi} \geq \mu\right)\right)\right)=1.$$ Hence, there is an assignment $\rho$ such that
 $$
 \begin{cases}
 (1)\quad \reals \rho \satfo \frac{ 1}{ 2} < \mu,\nu \leq  1\\[1mm]
 (2)\quad \reals \rho \satfo \bigwedge \Gamma\\[1mm]
 (3) \quad \reals \rho \nsatfo P^v_{\psi} \geq \mu.
 \end{cases}
 $$ 
 Consider the interpretation $(v,\rho)$ of $\UCL$. Then, by (2) $(v,\rho) \satuc \Gamma$. Moreover, by (3),
 $(v,\rho) \nsatuc \psi$.
\end{proof}

\paragraph{Complexity}\ \\[1mm]
We start by showing that the decision problem  $\ENTA^\UCL$ is in PSPACE. With an additional restriction on the size in bits of the number of unreliable connectives we are able to show that it is a co-NP complete problem. 

\begin{pth} \em
The problem $\ENTA^\UCL$ is in PSPACE.
\end{pth}
\begin{proof}
We are going to prove that algorithm in Figure~\ref{alg:enta} uses a polynomial amount of space. 
Let $n$ be the size in bits of the inputs $\psi$ and $\Gamma$. Thus, $|\var(\psi)|$, $|\Gamma|$, $|\fau(\psi)|$ are in $O(n)$.
Then:
\begin{itemize}
\item The storage of $v$ in Step 1 uses $O(n)$ bits;
\item The assignment in Step 1(a) uses $O(1)$ bits;
\item The storage of $\varphi$ in Step 1(b) uses $O(n)$ bits;
\item The inner cycle 1(b) iterates $O(2^n)$ times;
\item The algorithm $\cA_{\SATV^\PL}$ runs in polynomial space $O(n^k)$ for some $k$;
\item The storage of each coefficient of polynomial $\Pnu[\psi \triangleright \varphi]$ uses $O(n^k)$ bits for some $k$. Indeed
$$\Pnu[\psi \triangleright \varphi] = \sum_{j=0}^{|\fau(\psi)|-\ell} \left(\begin{array}{c} |\fau(\psi)|-\ell\\ j \end{array}\right) (-1)^j \nu^{\ell+j}$$
assuming that the number of unreliable connectives in $\psi$ that were replaced in $\varphi$ by ideal connectives is $\ell$;
\item The storage of polynomial $\Pnu[\psi \triangleright \varphi]$ uses $O(n^k)$ bits for some $k$;
\item Each coefficient in $P^v_{\psi}$ is obtained by summing  the coefficients of 
the same degree in $\{\Pnu[\psi \triangleright \varphi]: \varphi \in \Omega_\psi \text{ and } v \sat \varphi\}$. There are at most $2^n$ coefficients of the same degree each using $O(n^k)$ bits for some $k$ and so the maximum value of each coefficient is in $O(2^{n^k})$ for some $k$. Hence,
the maximum value of each coefficient of  $P^v_{\psi}$ is less than 
$$\sum_{j=1}^{2^n} 2^{n^k}= 2^{n^k+n}.$$
So, the storage of each coefficient in $P^v_{\psi}$ is in $O(n^k)$ for some $k$;
\item The storage of polynomial $P^v_{\psi}$ uses $O(n^k)$ bits for some $k$;
\item The number of operations when executing $\cA_{\SAT^{\exists\ORCF}}$ is
$$(|\Gamma|+5)^3 |\fau(\psi)|^{O(2)}$$
in the ring generated by the coefficients of the polynomials in
$$\exists  \mu\exists \nu \lneg \left(\frac{ 1}{ 2} < \mu,\nu \leq  1 \limp \left(\bigwedge \Gamma \limp 
    P^v_{\psi} \geq \mu\right)\right);$$
\item Each such ring operation executed by $\cA_{\SAT^{\exists\ORCF}}$ has a polynomial cost in terms of bit operations (see Section~8.1 of~\cite{bas:06}). Hence, the execution of $\cA_{\SAT^{\exists\ORCF}}$ in step 1(c) uses at most $O(n^k)$ bits for some $k$;
\item The number of bits used in each iteration of cycle 1 is the sum of the number of bits used in the storage of polynomial $P^v_{\psi}$
plus the number of bits used by $\cA_{\SAT^{\exists\ORCF}}$ in step 1(c). So it is $O(n^k)$ for some $k$;
\item Each iteration of cycle 1 reuses the space used in the previous iteration. So the cycle uses at most $O(n^k)$ bits for some $k$.
\end{itemize}
Thus, the space complexity of $\cA_{\ENTA^\UCL}(\psi,\Gamma)$ is in PSPACE. 
\end{proof}
We now consider a restricted version of the ambition constrained validity problem which is co-NP complete. Let 
$$\ENTA^\UCL_{\log}: \Lc_{\log}(X) \times \wp_{\text{fin}} \La\to \{0,1\}$$
be a map that given a formula $\psi$ representing a circuit with unreliable gates such that
 $|\fau(\psi)| \in O(\log \norm{\psi})$ 
and a finite set $\Gamma$ of ambition formulas, returns $1$ if and only if $\Gamma \entuc \psi$.
To analyse the complexity of $\ENTA^\UCL_{\log}$, we need to refer to the well known validity problem 
$$\VAL: L(X) \to \{0,1\}$$ that given a classical propositional formula $\varphi$ returns $1$ if and  only if  $\varphi$ is a tautology.
Observe that this problem is co-NP complete. 

\begin{pth} \em
The problem $\ENTA^\UCL_{\log}$  is co-NP complete. 
\end{pth}
\begin{proof} Indeed:\\
(a)~$\ENTA^\UCL_{\log}$ is in co-NP. We show that the complement of $\ENTA^\UCL_{\log}$ is in NP. That is, given $\psi$ and $\Gamma$
it returns $1$ iff $\Gamma \nentuc \psi$. We consider a  version of the algorithm  in Figure~\ref{alg:enta} that generates nondeterministically a valuation (the witness in this algorithm) and that returns $1$ whenever the original algorithm for that valuation returns $0$ and returns $0$ when the original algorithm for that valuation does not return $0$. Observe that the time for generating each valuation is polynomial on $\norm{\psi}$. 
Since the number of unreliable gates is logarithmic on $\norm{\psi}$, the inner cycle over $\Omega_\psi$ only iterates a polynomial number of times on $\norm{\psi}$. Moreover, each iteration
only takes at most  polynomial time on $\norm{\psi}$. 
The execution of $\cA_{\SAT^{\exists\ORCF}}$ in step 1(c)  takes $O(n^k)$ time, for some $k$. Hence, the whole verification part of the nondeterministic algorithm runs in polynomial time.\\[1mm]
(b)~Every problem in co-NP is reducible many-to-one in polynomial time to $\ENTA^\UCL_{\log}$.
Consider the map that
given a  formula $\varphi$ without unreliable connectives returns the pair $(\varphi,\emptyset)$. It is obvious to see that this map
is computable in polynomial time. Moreover, 
$$\VAL(\varphi)=1  \quad \text{iff} \quad \ENTA^\UCL_{\log}(\varphi,\emptyset)=1.$$
The thesis follows since $\VAL$ is co-NP complete. 
\end{proof}

\begin{corollary}\em
The problem $\ENTA^\UCL$  is co-NP hard. 
\end{corollary}
\begin{proof} Indeed $\ENTA^\UCL_{\log} \subset \ENTA^\UCL$ and,  by the previous theorem, $\ENTA^\UCL_{\log}$ is co-NP complete. 
\end{proof}


\paragraph{Application scenario} \ \\
Suppose that one wants to certify at least with probability equal to the reliability rate of the gates, that circuit $x_1 \lfdisj x_2$ with the unreliable connective $\lfdisj$ behaves as its ideal counterpart. Algorithm  $\cA_{\ENTA^\UCL}$ can be used for this verification. In fact
$$\cA_{\ENTA^\UCL}((x_1 \lfdisj x_2) \leqv (x_1 \ldisj x_2),\{\mu \leq \nu\})$$
returns 1 meaning that $\mu \leq \nu \entuc  (x_1 \lfdisj x_2) \leqv (x_1 \ldisj x_2)$ as one wants to show. 

\section{Model witness problem}\label{sec:mcp}

We now concentrate on the problem of constructing an interpretation (if there is at least one) satisfying a formula representing a circuit with unreliable gates. 
\begin{figure}
	       \hrule  \vspace*{4mm}
	       Input: c-formula $\psi$.
\begin{enumerate}
  \item For each $v \in V_{\var(\psi)}$ do:
  \begin{enumerate}
    \item Let $P^v_{\psi} := 0$;
    \item For each $\varphi \in \Omega_\psi$:
    \begin{enumerate}
       \item If $\SATV^\PL(\varphi,v)=1$ then 
       $P^v_{\psi}:=P^v_{\psi}+ \Pnu[\psi \triangleright \varphi]$;
    \end{enumerate}
       \item If $\tv{P^v_{\psi}}^{\nu \mapsto 1}>\frac12$ then Return $(1,(v,(1,\tv{P^v_{\psi}}^{\nu \mapsto 1})))$;
   \item If $\cA_{\SAT^{\exists\ORCF}}\left(\exists \nu \left((\frac{ 1}{ 2} <\nu <  1)\lconj(\frac{ 1}{ 2} <P^v_{\psi} <  1)\right)\right)=0$ then\\ Go to $1$;
    \item Let ${\overline\nu}_2:=3$;
    \item While True do:
    \begin{enumerate}
       \item For each ${\overline\nu}_1 \in \nats$ from $\lceil\frac{{\overline\nu_2+1}}{2}\rceil$ to $({\overline\nu}_2-1)$ do:
       \begin{enumerate}
   \item If $\tv{P^v_{\psi}}^{\nu \mapsto  \frac{{\overline\nu}_1}{{\overline\nu}_2}}>\frac12$ then Return $(1,(v,(\frac{{\overline\nu}_1}{{\overline\nu}_2},\tv{P^v_{\psi}}^{\nu \mapsto \frac{{\overline\nu}_1}{{\overline\nu}_2}})))$;
       \end{enumerate}
          \item Let ${\overline\nu}_2:={\overline\nu}_2+1$;
    \end{enumerate}

  \end{enumerate}
  \item Return $(0,\cdot)$.
  \end{enumerate}
       \hrule  \vspace*{4mm}
\caption{Algorithm $\cA_{\PMC^\UCL}$ for Problem $\PMC$.}
\label{alg:pmc}
\end{figure}

The {\it model witness problem} is a map  
$$\PMC^\UCL: \Lc(X) \to \left\{0,1\right\} \times \left(V \times \left(\rats \cap (\frac12,1]\right)^2\right)$$
where $V$ is the set of all finite valuations such that
\begin{itemize}
\item $\PMC(\psi)=(1,(v,(\ovl \nu, \ovl \mu)))$  implies that $(v,\rho) \satuc \psi$ where $\rho(\nu)= \ovl \nu$ and $\rho(\mu)= \ovl \mu$;
\item $\PMC(\psi)=(0,\cdot)$ implies that there are no $v$ and $\rho$  such that $(v,\rho) \satuc \psi$ and  $\frac 12 <  \rho(\mu),\rho(\nu)\leq 1$.\footnote{We use $\cdot$ for representing arguments not relevant for the case at hand.}
\end{itemize}

In order to propose an algorithm for the model witness problem, we need to refer to the $\SAT^{\exists\ORCF}$ problem described in the previous section.
\paragraph{Soundness and completeness}\ \\[1mm]
Let $\cA_{\PMC^\UCL}$ be the algorithm in Figure~\ref{alg:pmc}.  The fact that it is an algorithm follows from the next result. 

\begin{pth} \em
Let $\psi \in \Lc(X)$. Then,
\begin{itemize}
\item If $\PMC(\psi)=(1,\cdot)$ then $\cA_{\PMC^\UCL}(\psi)$ returns $(1,\cdot)$;
\item If $\PMC(\psi)=(0,\cdot)$ then $\cA_{\PMC^\UCL}(\psi)$ halts;
\item If $\cA_{\PMC^\UCL}(\psi)$ returns $(1,(v,(\ovl \nu, \ovl \mu)))$ then $(v,\rho)$,  where $\rho(\nu)= \ovl \nu$, $\rho(\mu)= \ovl \mu$,  is an interpretation satisfying $\psi$.
\end{itemize}
\end{pth} 
\begin{proof}\ \\
(1)~Assume that $\PMC(\psi)=(1,(v,(\ovl \nu, \ovl \mu)))$. 
Let $\rho$ be such that $\rho(\nu)= \ovl \nu$ and $\rho(\mu)= \ovl \mu$.
Then $(v,\rho) \satuc \psi$, that is 
$$\reals \rho \satfo \sum_{\begin{subarray} {c} \varphi\in\Omega_\psi\\
                                               v \sat \varphi\end{subarray}}\Pnu[\psi \triangleright \varphi]\geq \mu.$$
Observe that, when executing $\cA_{\PMC^\UCL}$ for $\psi$, 
$$P^v_\psi\text{ is }  \sum_{\begin{subarray} {c} \varphi\in\Omega_\psi\\
                                               v \sat \varphi\end{subarray}}\Pnu[\psi \triangleright \varphi] .$$
There are three cases to consider:\\
(a)~$\tv{P^v_\psi}^{\nu \mapsto 1}> \frac 12$. Then, $\cA_{\PMC^\UCL}(\psi)$  returns $(1,\cdot)$.
\\[2mm]
(b)~$\tv{P^v_\psi}^{\nu \mapsto 1}\leq  \frac 12$ and $\tv{P^v_\psi}^{\nu \mapsto \ovl \nu} < 1$. 
Then, in step (d), 
$$\cA_{\SAT^{\exists\ORCF}}\left(\exists \nu \left((\frac{ 1}{ 2} <\nu <  1)\lconj(\frac{ 1}{ 2} <P^v_{\psi} < 1)\right)\right)=1.$$
Let $$\mathcal P =\{2 \nu -1, 1-\nu, 2 P^v_\psi -1, 1-P^v_\psi \}.$$
Then $$S=\{d \in \reals: P(d) >0, P \in \mathcal P\}\neq \emptyset$$
and so, by Theorem~13.16 of~\cite{bas:06}, in each semi-algebraically connected component of $S$ there exists a rational number. 
Let $\frac m n$ be such a rational number. Then consider the assignment in (i) with $\ovl \nu_2=n$. Observe that
$$\ovl \nu_1 \in \{\lceil{\frac{n+1}{2}}\rceil,\dots,n-1\}$$
and that $2 \frac mn -1 >0$. Hence $m >\frac n 2$ and so $m \geq \lceil{\frac{n+1}{2}}\rceil$. On the other hand, 
$1 - \frac mn >0$ and so $m \leq n-1$. Then in cycle (i) there is an iteration where $\ovl \nu_1$ is $m$ and so,
since $2\tv{P^v_\psi}^{\nu \mapsto \frac {\ovl \nu_1}{\ovl \nu_2}}-1>0$  then the algorithm returns $(1,\cdot)$. \\[2mm]
(c)~$\tv{P^v_\psi}^{\nu \mapsto 1} \leq  \frac 12$ and $\tv{P^v_\psi}^{\nu \mapsto  \ovl \nu} = 1$. Then $\ovl \nu <1$. Hence, there is $\frac 12 <{\ovl \nu'} <1$  
such that $\frac 12 < \tv{P^v_\psi}^{\nu \mapsto  \ovl \nu'} < 1$. Then, this case follows (b).\\[2mm]
(2)~Assume that $\PMC(\psi)=(0,\cdot)$. Then, there are no $v$ and $\rho$ such that $(v,\rho) \satuc \psi$ and
$\frac 12 < \rho(\mu),\rho(\nu) \leq 1$. Then the guard in step 1(d) is always true and so the cycle in step 1(f) is never executed. Hence
the execution terminates since the set $V_{\var(\psi)}$ is finite.\\[2mm]
(3)~Assume that $\cA_{\PMC^\UCL}(\psi)$ returns $(1,(v,(\ovl \nu, \ovl \mu)))$. Then, one of following two cases holds:\\[1mm]
(a)~Condition in step 1(c) is satisfied for valuation $v$. Take $\rho(\nu)=1$ and $\rho(\mu)=\tv{P^v_{\psi}}^{\nu \mapsto 1}$. Then, $(v,\rho) \satuc \psi$. 
\\[1mm]
(b)~Condition in step 1(f)(A)  is satisfied for valuation $v$ and $\nu$ equal to $\frac {\ovl \nu_1}{\ovl \nu_2}$. Then take $\rho(\nu)=\frac {\ovl \nu_1}{\ovl \nu_2}$ and $\rho(\mu)=\tv{P^v_\psi}^{\nu \mapsto  \frac {\ovl \nu_1}{\ovl \nu_2}}$. Then, $(v,\rho) \satuc \psi$.
\end{proof}

\paragraph{Complexity}\ \\[1mm]
We start by showing that problem $\PMC$ is in PSPACE. Moreover, we also show that its decision version is NP complete when restricting 
the number of unreliable connectives in a c-formula.
 
\begin{pth} \em
Algorithm ${\cal A}_\PMC$ runs deterministically in polynomial space.
\end{pth}
\begin{proof} Let $\psi \in \Lc(X)$. 
We are going to prove that algorithm in Figure~\ref{alg:pmc} uses a polynomial amount of space. 
Let $n$ be the size in bits of the input $\psi$. Thus, $|\var(\psi)|$ and $|\fau(\psi)|$ are in $O(n)$.
Then:
\begin{itemize}
\item The storage of $v$ in Step 1 uses $O(n)$ bits;
\item The storage of $\varphi$ in Step 1(b) uses $O(n)$ bits;
\item The inner cycle 1(b) iterates $O(2^n)$ times;
\item The algorithm $\SATV^\PL$ runs in polynomial space $O(n^k)$ for some $k$;
\item The storage of each coefficient of polynomial $\Pnu[\psi \triangleright \varphi]$ uses $O(n^k)$ bits for some $k$;
\item The storage of polynomial $\Pnu[\psi \triangleright \varphi]$ uses $O(n^k)$ bits for some $k$;
\item The storage of each coefficient in $P^v_{\psi}$ is in $O(n^k)$ for some $k$;
\item The storage of polynomial $P^v_{\psi}$ uses $O(n^k)$ bits for some $k$;

\item The number of operations when executing $\cA_{\SAT^{\exists\ORCF}}$ is
$$4^2 |\fau(\psi)|^{O(1)}$$
in the ring generated by the coefficients of the polynomials in
$$\exists \nu \left((\frac{ 1}{ 2} <\nu <  1)\lconj(\frac{ 1}{ 2} <P^v_{\psi} <  1)\right);$$
\item Each such ring operation executed by $\cA_{\SAT^{\exists\ORCF}}$ has a polynomial cost in terms of bit operations (see Section~8.1 of~\cite{bas:06}). Hence, the execution of $\cA_{\SAT^{\exists\ORCF}}$ in step 1(d) uses at most $O(n^k)$ bits for some $k$;

\item The number of bits used in each iteration of cycle 1(f) is the sum of the number of bits used for storing $\ovl \nu_1$, $\ovl \nu_2$
and $\tv{P^v_\psi}^{\nu \mapsto \frac {\ovl \nu_1}{\ovl \nu_2}}$. By Theorem~13.16 of~\cite{bas:06}, the bitsize of $\frac{\ovl \nu_1}{\ovl \nu_2}$ is $\tau d^{O(1)}$ where $\tau$ is the bitsize of the coefficients of the polynomial $P^v_\psi$ of degree at most $d$. Hence, the bitsize of $\frac{\ovl \nu_1}{\ovl \nu_2}$ is $O(n^k)$ for some $k$. Therefore, the number of bits for storing $\ovl \nu_1$, $\ovl \nu_2$ is $O(n^k)$ for some $k$. Moreover, the number of bits for storing $\tv{P^v_\psi}^{\nu \mapsto \frac {\ovl \nu_1}{\ovl \nu_2}}$ is also polynomial;
\item Each iteration of cycle 1(f) reuses the space used in the previous iteration. So the cycle uses at most $O(n^k)$ bits for some $k$.
\end{itemize}
Thus, the algorithm runs deterministically in polynomial space. 
\end{proof}

It is straightforward to define the decision version of this problem. Such a problem is the satisfiability problem SAT for $\UCL$. 
Herein, we consider a restricted version of this problem which is NP complete.
The {\it logarithmic satisfiability problem} for $\UCL$ is the map 
$$\SATU^\UCL_{\log}: \Lc_{\log}(X) \to \{0,1\}$$
that given a formula $\psi$ representing a circuit with unreliable gates such that
 $|\fau(\psi)| \in O(\log (|\psi|))$, returns $1$ if and only if there is an interpretation $I$ such that $I \satuc \psi$.
To analyze the complexity of $\SATU^\UCL_{\log}$ we need to refer to the well known satisfiability problem for propositional logic. 
Let 
$$\text{SAT}: L(X) \to \{0,1\}$$ be the map that given a classical propositional formula $\varphi$ returns $1$ if and  only if  there is a valuation $v$ that satisfies $\varphi$.
Observe that the SAT problem is NP complete. 

\begin{figure}
	       \hrule  \vspace*{4mm}
	       Input: c-formula $\psi$ and valuation $v$.
  \begin{enumerate}
    \item Let $P^v_{\psi} := 0$;
    \item For each $\varphi \in \Omega_\psi$:
    \begin{enumerate}
       \item If $\SATV^\PL(\varphi,v)=1$ then 
       $P^v_{\psi}:=P^v_{\psi}+ \Pnu[\psi \triangleright \varphi]$;
    \end{enumerate}
       \item If $\tv{P^v_{\psi}}^{1}>\frac12$ then Return $1$;
   \item Return $\cA_{\SAT^{\exists\ORCF}}\left(\exists \nu \left((\frac{ 1}{ 2} <\nu <  1)\lconj(\frac{ 1}{ 2} <P^v_{\psi} <  1)\right)\right)$.
  \end{enumerate}
       \hrule  \vspace*{4mm}
\caption{Algorithm $\cA_{\SATU^\UCL_{\log}}$ for Problem $\SATU^\UCL_{\log}$.}
\label{alg:pmcsat}
\end{figure}

\begin{pth} \em
The problem $\SATU^\UCL_{\log}$  is NP complete.
\end{pth}
\begin{proof} Indeed:\\
(a)~$\SATU^\UCL_{\log}$ is in NP.  We consider the algorithm that  generates nondeterministically a valuation (the witness) 
with the verification subalgorithm  described in Figure~\ref{alg:pmcsat}. Observe that the generation of each valuation is polynomial on $\norm{\psi}$. 
Since the number of unreliable gates is logarithmic on $\norm{\psi}$, the cycle in step 2 over $\Omega_\psi$ only iterates a polynomial number of times. Moreover, each iteration
only takes polynomial time. 
The execution of $\cA_{\SAT^{\exists\ORCF}}$ takes $O(\norm{\psi}^k)$ time, for some $k$. Moreover,
the evaluation of a polynomial of degree $O(\log(\norm{\psi}))$ on a rational constant is polynomial in time. 
Hence, the whole verification part of the nondeterministic algorithm runs in polynomial time.\\[1mm]
(b)~Every problem in NP is reducible many-to-one in polynomial time to $\SATU^\UCL_{\log}$.
Consider the identity map over $L(X)$. It is computable in polynomial time. Since, 
$$\text{SAT}(\varphi)=1  \quad \text{iff} \quad \SATU^\UCL_{\log}(\varphi)=1$$
using Proposition~2.1 in \cite{acs:jfr:css:pmat:13}, 
the thesis follows taking into account that $\text{SAT}$ is NP complete. 
\end{proof}

\paragraph{Application scenario} \ \\
The algorithm $\cA_{\PMC^\UCL}(\psi)$ is stronger than a model checking algorithm since the former builds explicitly a model when  the given formula is satisfiable. Recall that in a model checking algorithm a model is also given and the objective is to verify 
whether or not the model satisfies the formula. 

Given a circuit, $\cA_{\PMC^\UCL}(\psi)$ allows 
to determine: (1)~a target success rate, (2)~a reliability rate of the gates and (3)~a classical model, satisfying such a circuit for the returned success and reliability rates. 

For example, consider the application of algorithm $\cA_{\PMC^\UCL}(\psi)$ to the formula $$((\lfneg x) \ldisj  (\lfneg x) \ldisj (\lfneg x)) \leqv (x \ldisj (\lneg x)).$$ There are several cases to consider.
\begin{itemize}
\item Assume that the first valuation considered in cycle 1 is $v$ such that $v(x)=0$. Then $P^v_\psi$ is
the polynomial $1-(1-\nu)^3$ where $\psi$ is the formula above. Then $\tv{P^v_\psi}^{\nu \mapsto 1}=1 > \frac 12$ and so the algorithm returns
$$(1,(v,(1,1))$$
in step 1(c).

\item Assume that the first valuation considered in cycle 1 is $v$ such that $v(x)=1$.  Then $P^v_\psi$ is
the polynomial $1-\nu^3$ where $\psi$ is the formula above. Then $\tv{P^v_\psi}^{\nu \mapsto 1}=0 \not > \frac 12$. 
Then cycle 1(f) returns immediately $(1,(v, \frac 2 3, \frac{19}{27}))$. 
\end{itemize}
\section{Reliability rate abduction problem} \label{sec:rra}

In this section, we discuss the problem of finding possible intervals, of length defined by  a given natural number, where the value of $\nu$ guarantees a given target success rate of a given  circuit. 
 
The {\it reliability rate abduction problem} is the map
$$\textstyle \ARR^\UCL: \Lc(X) \times \left(\rats \cap \left(\frac{ 1}{ 2},1\right]\right)  \times \nats^+\to \wp_{\text{fin}} ((\rats \cap \left(\frac{ 1}{ 2},1\right])\times (\rats \cap \left(\frac{ 1}{ 2},1\right] ))$$
that given a  formula $\psi$ representing a circuit with unreliable gates, a rational number $\ovl \mu$ in the interval $\left(\frac{ 1}{ 2},1\right]$ and a positive natural number $k$, returns the set $\ell$ of all intervals of the form $(\frac12 +  \frac j{2k}, \frac12 +  \frac{j+1}{2k}]$ for $j \in \{0,\dots,k-1\}$ such that for every interval $(a,b] \in \ell$, valuation $v$  and assignment $\rho$ over $\reals$ such that $\rho(\mu)=\ovl \mu$ and $a < \rho(\nu) \leq b$, 
$$(v,\rho) \satuc \psi.$$

\begin{figure}
	       \hrule  \vspace*{4mm}
	       Inputs: c-formula $\psi$, $\ovl \mu$ in $\rats \cap (\frac{ 1}{ 2},1]$ and $k \in \nats^+$.
\begin{enumerate}
\item $\ell :=\{\}$;
    \item For each $j=0,\dots, k-1$:
    \begin{enumerate}
  \item For each $v \in V_{\var(\psi)}$ do:
  \begin{enumerate}
    \item Let $P^v_{\psi} := 0$;
    \item For each $\varphi \in \Omega_\psi$:
    \begin{enumerate}
       \item If $\SATV^\PL(\varphi,v)=1$ then 
       $P^v_{\psi}:=P^v_{\psi}+ \Pnu[\psi \triangleright \varphi]$;
    \end{enumerate}
        \item If $\SAT^{\exists\ORCF}\left(\exists \nu \left(\left (\frac12 +  \frac j{2k} < \nu \leq \frac12 +  \frac{j+1}{2k}\right) \lconj 
    (\ovl \mu  > P^v_{\psi})
    \right)\right)=1$\\ then go to 2;
  \end{enumerate}
  \item $\ell:= \ell \cup \{(\frac12 +  \frac j{2k} ,\frac12 +  \frac{j+1}{2k})\}$
  \end{enumerate}
    \item Return $\ell$;
\end{enumerate}
       \hrule  \vspace*{4mm}
\caption{Algorithm $\cA_{\ARR^\UCL}$ for Problem $\ARR^\UCL$.}
\label{alg:arr}
\end{figure}

\paragraph{Soundness and completeness}\ \\
Let $\cA_{\ARR^\UCL}$ be the algorithm in Figure~\ref{alg:arr}. Observe that it is in fact an algorithm since 
(1)~the cycles in Step 2, Step 2(a) and Step 2(a)(ii) are over finite sets $\{0,\dots,k-1\}$,  $V_{\var(\psi)}$ and $\Omega(\psi)$, respectively; and 
(2)~$\SATV^\PL$ and $\SAT^{\exists\ORCF}$ are algorithms. 

\begin{pth} \em
Let $\psi \in \Lc(X)$, $\ovl \mu \in \rats \cap (\frac{ 1}{ 2},1]$, and $k \in \nats^+$. Then,
$${\cA_{\ARR^\UCL}}(\psi,\ovl \mu,k) \text{ returns } \ell \quad \text{iff} \quad \ARR^\UCL(\psi,\ovl \mu,k)=\ell.$$
\end{pth} 
\begin{proof} \ \\
($\to$)~Assume that  $\cA_{\ARR^\UCL}(\psi,\ovl \mu,k)$  returns $\ell$. 
Let $(\frac12 +  \frac j{2k} ,\frac12 +  \frac{j+1}{2k}) \in \ell$.
Then for every $\ovl \nu$ such that 
$$\left(\frac12 +  \frac j{2k} < \ovl\nu \leq \frac12 +  \frac{j+1}{2k}\right)$$  we have 
    $\tv{P^v_{\psi}}^{\nu \mapsto \ovl \nu} \geq \ovl \mu\; \text{ for each } v$.
Let $\rho$ be such that $\rho(\nu)\in (\frac12 +  \frac j{2k} ,\frac12 +  \frac{j+1}{2k})$ and  $\rho(\mu)=\ovl\mu$.
Then $(v,\rho) \satuc \psi$ for every valuation $v$ and so 
$(\frac12 +  \frac j{2k} ,\frac12 +  \frac{j+1}{2k}) \in \ARR^\UCL(\psi,\ovl \mu,k)$.\\[1mm]
$(\from)$~Let $(\frac12 +  \frac j{2k} ,\frac12 +  \frac{j+1}{2k}) \in \ARR^\UCL(\psi,\ovl \mu,k)$. Then,
for every valuation $v$ and assignment $\rho$ over $\reals$ such that $\rho(\mu)=\ovl \mu$ and $\frac12 +  \frac j{2k} < \rho(\nu) \leq \frac12 +  \frac{j+1}{2k}$ we have
$$(v,\rho) \satuc \psi,$$
that is, 
$$\reals \rho \satuc \sum_{\begin{subarray} {c} \varphi\in\Omega_\psi\\
                                               v \sat \varphi\end{subarray}}\Pnu[\psi \triangleright \varphi] \geq \mu.$$
Since
$$P^v_{\psi}=\sum_{\begin{subarray} {c} \varphi\in\Omega_\psi\\
                                               v \sat \varphi\end{subarray}}\Pnu[\psi \triangleright \varphi]$$
then
$$\reals \rho \satuc P^v_{\psi} \geq \mu.$$
Thus, when running algorithm ${\cA_{\ARR^\UCL}}$  with input $(\psi,\ovl \mu,k)$,  for every valuation, the guard of the If in step~2(a)(iii)  is false and so the interval 
$(\frac12 +  \frac j{2k} ,\frac12 +  \frac{j+1}{2k})$ is in the list returned by ${\cA_{\ARR^\UCL}}(\psi,\ovl \mu,k)$.
\end{proof}

\paragraph{Complexity}\ \\
We now analyze the space complexity of algorithm $\cA_{\ARR^\UCL}$. Afterwards we discuss the decision version of problem $\ARR^\UCL$. 

\begin{pth} \em
The algorithm $\cA_{\ARR^\UCL}$ runs deterministically in polynomial space.
\end{pth}
\begin{proof}
We are going to prove that algorithm in Figure~\ref{alg:arr} uses a polynomial amount of space.  Let $n$ be the size in bits of the input $\psi$, $\ovl \mu$ and $k$. Thus, $|\var(\psi)|$, $\norm{\ovl \mu}$, $\norm{k}$ and $|\fau(\psi)|$ are in $O(n)$.
Then:
\begin{itemize}
\item The storage of $j$ in Step 2 uses $O(n)$ bits;
\item The storage of $v$ in Step 2(a) uses $O(n)$ bits;
\item The storage of $\varphi$ in Step 2(a)(ii) uses $O(n)$ bits;
\item The inner cycle 2(a)(ii) iterates $O(2^n)$ times;
\item The algorithm $\SATV^\PL$ runs in polynomial space $O(n^k)$ for some $k$;
\item The storage of each coefficient of polynomial $\Pnu[\psi \triangleright \varphi]$ uses $O(n^k)$ bits for some $k$;
\item The storage of polynomial $\Pnu[\psi \triangleright \varphi]$ uses $O(n^k)$ bits for some $k$;
\item The storage of each coefficient  of  $P^v_{\psi}$  uses $O(n^k)$ bits for some $k$;
\item The storage of polynomial $P^v_{\psi}$ uses $O(n^k)$ bits for some $k$;

\item The number of operations when executing $\SAT^{\exists\ORCF}$ is
$$25 |\fau(\psi)|^{O(1)}$$
in the ring generated by the coefficients of the polynomials in
$$\exists \nu \left(\left (\frac12 +  \frac j{2k} < \nu \leq \frac12 +  \frac{j+1}{2k}\right) \lconj 
    (\ovl \mu  > P^v_{\psi})
    \right);$$
\item Each such ring operation executed by $\SAT^{\exists\ORCF}$ has a polynomial cost in terms of bit operations (see Section~8.1 of~\cite{bas:06}). Hence, the execution of $\SAT^{\exists\ORCF}$ in step 2(a)(iii) uses at most $O(n^k)$ bits for some $k$;
\end{itemize}
Thus, the algorithm runs deterministically in polynomial space.
\end{proof}

We now consider a decision problem for $\ARR^\UCL$ and prove that it is NP complete.
The {\it logarithmic reliability rate decision problem} is the map 
$$\RRD^\UCL_{\log}: \Lc_{\log}(X) \times \left(\rats \cap \left(\frac{ 1}{ 2},1\right]\right) \to \{0,1\}$$
that given a formula $\psi$ representing a circuit with unreliable gates such that
 $|\fau(\psi)|, |\var(\psi)| \in O(\log \norm{\psi})$ and a success circuit rate $\ovl \mu$, returns $1$ if and only if there is an assignment $\rho$ with 
 $\rho(\mu)=\ovl \mu$ such that
$$(v,\rho) \satuc \psi$$ for every valuation $v$.

\begin{figure}
	       \hrule  \vspace*{4mm}
	       Inputs: c-formula $\psi$ and $\ovl \mu$ in $\rats \cap (\frac{ 1}{ 2},1]$.
\begin{enumerate}
\item $\ell := \{\}$;
  \item For each $v \in V_{\var(\psi)}$ do:
  \begin{enumerate}
    \item Let $P^v_{\psi} := 0$;
    \item For each $\varphi \in \Omega_\psi$:
    \begin{enumerate}
       \item If $\SATV^\PL(\varphi,v)=1$ then 
       $P^v_{\psi}:=P^v_{\psi}+ \Pnu[\psi \triangleright \varphi]$;
    \end{enumerate}
    \item $\ell :=\ell \cup \{P^v_{\psi}\}$;
  \end{enumerate}
        \item Return $\SAT^{\exists\ORCF}\left(\exists \nu \left(\left (\frac12  < \nu \leq 1\right) \lconj 
    \bigwedge_{P \in \ell} \left(\ovl \mu  \leq P\right)
    \right)\right)$.
\end{enumerate}
       \hrule  \vspace*{4mm}
\caption{Algorithm $\cA_{\RRD^\UCL}$ for Problem $\RRD^\UCL$.}
\label{alg:rrd}
\end{figure}

\begin{pth} \em
The problem $\RRD^\UCL_{\log}$  is in $P$.
\end{pth}
\begin{proof} Indeed:\\
We are going to prove that algorithm in Figure~\ref{alg:rrd} runs in polynomial time.  Let $n$ be the size in bits of the input $\psi$ and $\ovl \mu$. Thus, $|\var(\psi)|$ and $|\fau(\psi)|$ are in $O(\log n)$.
Then:
\begin{itemize}
\item The  cycle in step 2 iterates $O(n)$ times;
\item The inner cycle in step 2(b) iterates $O(n)$ times;
\item The algorithm $\SATV^\PL$ runs in polynomial time $O(n^k)$ for some $k$;
\item The computation of $\Pnu[\psi \triangleright \varphi]$ runs in $O(n^k)$ for some $k$;
\item The sum of the polynomials $\Pnu[\psi \triangleright \varphi]$ and $P^v_{\psi}$  takes  $O(n^k)$ time for some $k$;
\item The time complexity of executing cycle 2(b) is in $O(n^k)$ for some $k$;
\item The time complexity of executing cycle 2 is in $O(n^k)$ for some $k$;
\item The number of operations when executing $\SAT^{\exists\ORCF}$ is
$$O(n)^2O(\log n)^{O(1)}$$
in the ring generated by the coefficients of the polynomials in
$$\exists \nu \left(\left (\frac12  < \nu \leq 1\right) \lconj 
    \bigwedge_{P \in \ell} \left(\ovl \mu  \leq P\right)
    \right);$$
\item Each such ring operation executed by $\SAT^{\exists\ORCF}$ has a polynomial cost in terms of bit operations (see Section~8.1 of~\cite{bas:06}). Hence, the execution of $\SAT^{\exists\ORCF}$ in step 3 takes at most $O(n^k)$ time for some $k$.
\end{itemize}
Thus, the algorithm runs  in polynomial time.
\end{proof}

\paragraph{Application scenario} \ \\
Suppose that we want to determine the reliability rate of gates that ensure that
the circuit represented by the formula
$x \lfdisj (\lfneg x)$ has a success rate of $0.7$. We can use algorithm $\cA_{\ARR^\UCL}$ to address this problem. Indeed, using 
$k=3$, the execution of $\cA_{\ARR^\UCL}(x \lfdisj (\lfneg x),0.7,3)$ returns $\{(0.875,1)\}$. This means that if the circuit uses gates with a reliability rate over $0.875$ then the probability of $x \lfdisj (\lfneg x)$ being equivalent to $x \ldisj (\lneg x)$ is at least $0.7$.

\section{Success rate optimization problem}\label{sec:sro}

Assume that we have two circuits with unreliable gates and we want to determine how close they are of being equivalent by finding the reliability rate of the gates that maximizes their equivalence. In a more abstract way, this problem consists of, given a formula in $\UCL$, finding a success rate and a reliability rate of the gates in such a way that the success rate is a maximum. 

The {\it success rate optimization problem} is the map
$$\OSC^\UCL: \Lc(X) \to  \{0,1\} \times \left(\rats \cap \left(\frac12,1\right]\right)^2$$
that given a formula $\psi$ representing a circuit with unreliable gates returns either $1$ together with a pair $(\ovl \nu,\ovl\mu)$ of rational numbers such that:
\begin{itemize}
\item $(v,\rho) \satuc \psi$ for every valuation $v$ and assignment $\rho$ such that $\rho(\nu)=\ovl \nu$ and $\rho(\mu)=\ovl \mu$;
\item For every assignment $\rho'$ with $\frac 12 <\rho'(\nu),\rho'(\mu)\leq 1$, if $(v',\rho') \satuc \psi$ for every valuation $v'$ then $\rho'(\mu) \leq \ovl \mu$;
\end{itemize}
or returns $0$ as the first component if there is no such pair $(\ovl \nu,\ovl\mu)$ satisfying the above conditions.

In order to propose an algorithm for this problem, we need to introduce some material related to the theory of real closed ordered fields. 

Let $\cal P$ be a finite set of polynomials over $\Sigma_\orcf$. A $\cal P$-atom is a formula of the form $P=0$, $P \neq 0$, $P >0$ and $P <0$ where $P \in \cal P$, and a $\cal P$-formula is a formula written with $\cal P$-atoms.  Given a set $\cal P$ of polynomials, 
let $\textit{QPL}_{\cal P}(\Sigma_\orcf)$ be the set of all $\cal P$-formulas of the form 
$$Qx_1 \dots Qx_n \; \delta$$ where $Q$ is either $\forall$ or $\exists$ and  $\delta$ is a quantifier free formula, and let 
$\textit{QFL}(\Sigma_\orcf)$ be the set of all quantifier free formulas over $\Sigma_\orcf$.

Given a set of polynomials $\cal P$, the cylindrical quantifier elimination problem is the map 
$$\CQE_{\cal P}: \textit{QPL}_{\cal P}(\Sigma_\orcf) \to \textit{QFL}(\Sigma_\orcf)$$
that given a formula $\eta$ in $\textit{QPL}_{\cal P}(\Sigma_\orcf)$ returns $\lfalsum$ if $\eta$ is not satisfiable, otherwise 
returns 
a formula $\eta'$ in  $\textit{QFL}(\Sigma_\orcf)$
with the same free variables and such that $\reals \satfo \eta \leqv \eta'$. Herein, we consider the algorithm 11.16 of~\cite{bas:06} for solving this problem, that runs in 
$(sd)^{{O(1)}^k}$ time where $s$ is the number of polynomials in $\cal P$, $d$ is a bound on the degree of the polynomial and $k$ is the number of quantified variables in the given formula. 

\begin{figure}[ht]
	       \hrule  \vspace*{4mm}
	       Inputs: c-formula $\psi$.
\begin{enumerate}
\item ${\cal P} :=\{\}$;
  \item For each $v \in V_{\var(\psi)}$ do:
  \begin{enumerate}
    \item Let $P^v_{\psi} := 0$;
    \item For each $\varphi \in \Omega_\psi$:
    \begin{enumerate}
       \item If $\SATV^\PL(\varphi,v)=1$ then 
       $P^v_{\psi}:=P^v_{\psi}+ \Pnu[\psi \triangleright \varphi]$;
    \end{enumerate}
   \item ${\cal P}:= {\cal P} \cup \{P^v_{\psi}\}$;
  \end{enumerate}
    \item $\eta := \displaystyle \forall x \forall y  \left(\frac12 < \mu,\nu\leq 1 \lconj \bigwedge_{P \in {\cal P}} \mu \leq P \right) \; \lconj$\\[1mm]
                        \hspace*{20mm} $\displaystyle \left(\left(\frac12 < x,y \leq 1 \lconj \bigwedge_{P \in {\cal P}} y \leq [P]^\nu_x \right)\limp y \leq \mu\right)$;
    \item Return $\CQE_{{\cal P}}(\eta)$.
    \end{enumerate}
       \hrule  \vspace*{4mm}
\caption{Algorithm $\cA_{\OSC^\UCL}$ for Problem $\OSC^\UCL$.}
\label{alg:orcg}
\end{figure}

\paragraph{Soundness and completeness}\ \\[1mm]
Let $\cA_{\OSC^\UCL}$ be the algorithm in Figure~\ref{alg:orcg} where $[P]^\nu_x$ is the polynomial obtained from polynomial $P$ by replacing $\nu$ by $x$.
Observe that it is in fact an algorithm since 
(1)~the cycles in Step 2 and in Step 2(b) are over a finite set of valuations and outcomes, respectively; and (2)~$\SATV^\PL$ and $\CQE_{\cal P}$ are algorithms.

\begin{pth} \em
Let $\psi \in \Lc(X)$.  Then,
\begin{enumerate}
\item $\cA_{\OSC^\UCL}(\psi)$ returns  $\lfalsum$  if and only if  $\OSC^\UCL(\psi)|_1=0$;
\item If $\cA_{\OSC^\UCL}(\psi)  \; \text{returns a satisfiable formula} \;(\mu\cong \ovl \mu) \lconj \alpha$, where $\ovl \mu$ is a term without variables and $\alpha$ is a disjunction of a conjunction of atoms involving only variable $\nu$ then $\OSC^\UCL(\psi)=(1,(\ovl \nu,\ovl \mu))$ where
$\ovl \nu$ satisfies $\alpha$ in $\ORCF$;
\item If $\OSC^\UCL(\psi)=(1,(\ovl \nu,\ovl \mu))$ then $\reals \rho \satfo \cA_{\OSC^\UCL}(\psi)$ where $\rho$ is such that $\rho(\mu)=\ovl \mu$ and $\rho(\nu)=\ovl \nu$. 
\end{enumerate}
\end{pth} 
\begin{proof}\ \\
1.~$(\to)$~Assume that $\cA_{\OSC^\UCL}(\psi) \; \text{returns} \; \lfalsum$. Then there are two cases:\\
(a)~There are no $\ovl \mu$ and $\ovl \nu$ such that
$$\reals \rho \satfo \left(\frac12 < \mu,\nu\leq 1 \lconj \bigwedge_{P \in {\cal P}} \mu \leq P \right)$$
where $\rho$ is such that $\rho(\mu)=\ovl \mu$ and $\rho(\nu)=\ovl \nu$. Hence, for every assignment $\rho$ there is a valuation $v$  such that 
$$\reals \rho \nsatfo \left(\frac12 < \mu,\nu\leq 1 \lconj  \mu \leq P^v_\psi \right),$$
that is, $(v,\rho) \nsatuc \psi$. So $\OSC^\UCL(\psi)|_1$ is $0$.\\[1mm]
(b)~Otherwise, there is $\rho$ such that
$$\reals \rho \nsatfo \displaystyle \left(\left(\frac12 < x,y \leq 1 \lconj \bigwedge_{P \in {\cal P}} y \leq [P]^\nu_x \right)\limp y \leq \mu\right).$$
Hence there is not a solution of optimisation problem:
$$\begin{cases}
\displaystyle\max_{\mu,\nu} \;\mu\\[2mm]
\frac12 < \mu,\nu\leq 1 \lconj \bigwedge_{P \in {\cal P}} \mu \leq P
\end{cases}$$
and so $\OSC^\UCL(\psi)|_1$ is $0$.\\[2mm]
$(\from)$~$\OSC^\UCL(\psi)|_1=0$. There are two cases. \\
(a)~For every assignment $\rho$ with $\frac 12 < \rho(\mu),\rho(\nu)\leq 1$
there is a valuation $v$ such that $(v,\rho) \nsatuc \psi$. Hence, given such a $\rho$ let $v$ be such that $(v,\rho) \nsatuc \psi$.
Thus, $\reals \rho \nsatfo \mu \leq P^v_\psi$ and so $\CQE_{{\cal P}}(\eta)$ returns $\lfalsum$.\\[1mm]
(b)~Otherwise, for every assignment $\rho$ with $\frac 12 < \rho(\mu),\rho(\nu)\leq 1$ such that $(v,\rho) \satuc \psi$ for every valuation $v$, 
there is an assignment $\rho'$ with $\frac 12 < \rho'(\mu),\rho'(\nu)\leq 1$ such that $(v',\rho') \satuc \psi$ 
for every valuation $v'$ and $\rho'(\mu) > \rho(\mu)$. Let $\rho$ be such that $\frac 12 < \rho(\mu),\rho(\nu)\leq 1$ and $(v,\rho) \satuc \psi$ for every valuation $v$ and $\rho'$ with $\frac 12 < \rho'(\mu),\rho'(\nu)\leq 1$ such that $(v',\rho') \satuc \psi$ 
for every valuation $v'$ and $\rho'(\mu) > \rho(\mu)$. Take 
$\rho''$ as the assignment such that $\rho''(x)=\rho'(\nu)$,
$\rho''(y)=\rho'(\mu)$, $\rho''(\nu)=\rho(\nu)$ and $\rho''(\mu)=\rho(\mu)$. 
Hence
$$\reals \rho'' \nsatfo \displaystyle \left(\left(\frac12 < x,y \leq 1 \lconj \bigwedge_{P \in {\cal P}} y \leq [P]^\nu_x \right)\limp y \leq \mu\right)$$
and so $\CQE_{{\cal P}}(\eta)$ returns $\lfalsum$.\\[1mm]
2.~Assume that $\cA_{\OSC^\UCL}(\psi)$ returns a satisfiable formula $(\mu\cong \ovl \mu) \lconj \alpha$ 
where $\ovl \mu$ is a term without variables and $\alpha$ is a disjunction of a conjunction of atoms involving only variable $\nu$. Then, 
$$\begin{array}{l}
\reals \satfo \displaystyle \forall x \forall y  \left(\frac12 < \mu,\nu\leq 1 \lconj \bigwedge_{P \in {\cal P}} \mu \leq P \right) \lconj \\
\hspace*{20mm}\displaystyle \left(\left(\frac12 < x,y \leq 1 \lconj \bigwedge_{P \in {\cal P}} y \leq [P]^\nu_x \right)\limp y \leq \mu\right) \\[4mm]
 \hspace*{10mm} \leqv \\[2mm]
\hspace*{10mm} (\mu\cong \ovl \mu) \lconj \alpha.
\end{array}
$$
Since $(\mu\cong \ovl \mu) \lconj \alpha$ is a satisfiable formula let $\rho$ be an assignment such that $\rho(\mu)=\ovl \mu$ and
$\reals \rho \satfo \alpha$. Then, for each valuation $v$, $\reals \rho \satfo \mu \leq P_\psi^v$ and so $(v,\rho) \satuc \psi$ for every $v$.
Moreover, let $\rho'$ be an assignment with $\frac 12 <\rho'(\nu),\rho'(\mu)\leq 1$ and assume that $(v',\rho') \satuc \psi$ for every valuation $v'$. Take $\rho''$ as the assignment such that $\rho''(x)=\rho'(\nu)$,
$\rho''(y)=\rho'(\mu)$, $\rho''(\nu)=\rho(\nu)$ and $\rho''(\mu)=\rho(\mu)$. Hence
$$\begin{array}{l}
\reals \rho''\satfo \displaystyle \forall x \forall y  \left(\frac12 < \mu,\nu\leq 1 \lconj \bigwedge_{P \in {\cal P}} \mu \leq P \right) \lconj \\
\hspace*{20mm}\displaystyle \left(\left(\frac12 < x,y \leq 1 \lconj \bigwedge_{P \in {\cal P}} y \leq [P]^\nu_x \right)\limp y \leq \mu\right).
\end{array}
$$
Since
$$\reals \rho'' \satfo \frac12 < x,y \leq 1 \lconj \bigwedge_{P \in {\cal P}} y \leq [P]^\nu_x$$
then $$\reals \rho'' \satfo y \leq \mu.$$
Therefore $\rho'(\mu) \leq \rho(\mu)$ and so $\rho'(\mu) \leq \ovl \mu$. Thus, $\OSC^\UCL(\psi)=(1,(\ovl \nu,\ovl \mu))$ for
some $\ovl \nu$. \\[1mm]
3.~Assume that $\OSC^\UCL(\psi)=(1,(\ovl \nu,\ovl \mu))$. Let $\rho$ be an assignment such that 
$\rho(\mu)=\ovl \mu$ and $\rho(\nu)=\ovl \nu$. Then $(v,\rho) \satuc \psi$  hence 
$\reals \rho \satfo \mu \leq P_\psi^v$
for every valuation $v$ and so
$$\reals \rho \satfo \left(\frac12 < \mu,\nu\leq 1 \lconj \bigwedge_{P \in {\cal P}} \mu \leq P \right).$$
Let $\rho'$ be $\{x,y\}$-equivalent to $\rho$ with $\frac 12 < \rho'(x),\rho'(y) \leq 1$ and assume that 
$$\reals \rho' \satfo \bigwedge_{P \in {\cal P}} y \leq [P]^\nu_x.$$
Let $\rho''$ be an assignment such that $\rho''(\nu)=\rho'(x)$ and $\rho''(\mu)=\rho'(y)$. Hence
$$\reals \rho'' \satfo \bigwedge_{P \in {\cal P}} \mu \leq P.$$
Therefore, $(v',\rho'') \satuc \psi$ for very valuation $v'$. Thus, $\rho''(\mu) \leq \rho(\mu)$ and so
$\rho'(y) \leq \rho(\mu)$. Finally,
$$\reals \rho' \satfo \displaystyle \left(\left(\frac12 < x,y \leq 1 \lconj \bigwedge_{P \in {\cal P}} y \leq [P]^\nu_x \right)\limp y \leq \mu\right)$$
and so $\reals\rho \satfo \cA_{\OSC^\UCL}(\psi)$.
\end{proof}

\paragraph{Complexity} \ \\
We analyse the time complexity of algorithm $\cA_{\OSC^\UCL}$. 

\begin{pth} \em
The algorithm $\cA_{\OSC^\UCL}$ runs deterministically in exponential time.
\end{pth}
\begin{proof}
We are going to prove that algorithm in Figure~\ref{alg:orcg} takes an exponential amount of time.  Let $n$ be the size in bits of the input $\psi$. Thus, $|\var(\psi)|$ is in $O(n)$.
Then:
\begin{itemize}

\item The outer cycle in step 2 iterates $O(2^n)$ times;
\item The inner cycle in step 2(b) iterates $O(2^n)$ times;
\item The algorithm $\SATV^\PL$ runs in polynomial time $O(n^k)$ for some $k$;
\item The computation of $\Pnu[\psi \triangleright \varphi]$ runs in $O(n^k)$ for some $k$;
\item The sum of the polynomials $\Pnu[\psi \triangleright \varphi]$ and $P^v_{\psi}$  takes  $O(n^k)$ time for some $k$;
\item The time complexity of executing cycle in step 2(b) is in $O(n^k)$ for some $k$;
\item The time complexity of executing cycle in step 2 is in $O(n^k)$ for some $k$;

\item The number of operations when executing $\CQE$ is
$$(2^nn)^{O(1)^2}$$ since the number of polynomials is $O(2^n)$, the maximum degree of each polynomial is $O(n)$ and the number of variables is $2$. The operations are over the integral domain generated by the coefficients of the polynomials in\\[1mm]
$\displaystyle \forall x \forall y  \left(\frac12 < \mu,\nu\leq 1 \lconj \bigwedge_{P \in {\cal P}} \mu \leq P \right) \; \lconj$\\[1mm]
                        \hspace*{20mm} $\displaystyle \left(\left(\frac12 < x,y \leq 1 \lconj \bigwedge_{P \in {\cal P}} y \leq [P]^\nu_x \right)\limp y \leq \mu\right);$
\item Each such integral domain operation executed by $\CQE$ has a polynomial cost in terms of bit operations (see Section~8.1 of~\cite{bas:06}). Hence, the execution of $\CQE$ in step 3 takes at most $O(n^k)$ time for some $k$.
\end{itemize}
Hence, algorithm $\cA_{\OSC^\UCL}$ runs deterministically in exponential time.
\end{proof}

\paragraph{Application scenario} \ \\
Suppose that we want to determine the maximum success rate of a circuit and the reliability rate of the gates that ensures such a maximum.
In this case we can use algorithm $\cA_{\OSC^\UCL}$. For example, assume that we want to investigate how close is the circuit  represented by the following c-formula
$$(\lfneg x) \ldisj  (\lfneg x) \ldisj (\lfneg x)$$
to the circuit represented by the formula $x \ldisj (\lneg x)$. 
For that, we can apply algorithm $\cA_{\OSC^\UCL}$  to the formula
$$((\lfneg x) \ldisj  (\lfneg x) \ldisj (\lfneg x)) \leqv (x \ldisj (\lneg x))$$
which  returns
$$(\nu\cong 0.6306) \lconj (\mu \cong 0.931773).$$
This means that the maximum fidelity between the given circuits is $0.931773$ providing that the reliability rate of the gates is  $0.6306$.
This information is also useful to an engineer in what concerns the gates to be used in order to obtain the circuit with the best possible success rate.

\section{Outlook}\label{sec:outlook}

In this paper, we investigated the complexity of some computational and decision problems in the logic $\UCL$ for reasoning about circuits with unreliable gates  (cf~\cite{acs:jfr:css:pmat:13}).~The algorithms herein proposed for these problems rely on algorithms that are known for propositional logic (which was expected since $\UCL$ is a conservative extension of $\PL$) and on algorithms for real closed ordered fields (see~\cite{bas:06}). 

Our results allow to conclude that establishing the validity in $\UCL$ is not more difficult than to establish validity in propositional logic, providing that
the number of connectives representing unreliable gates is logarithmic on the length of the formula. Moreover, the same happens with respect to the satisfiability problem. That is, establishing the satisfiability in $\UCL$ is not more difficult than to establish satisfiability in propositional logic, providing that
the number of connectives representing unreliable gates is logarithmic on the length of the formula.

Research on $\UCL$ aims at extending the logic to allowing probabilistic inputs and afterwards to address the issue of 
reasoning about quantum circuits with unreliable gates. Thereafter, we intend to analyze computation and decision problems 
in this context.

\section*{Acknowledgments}

This work was supported by Fundação para a Ciência e a Tecnologia by way of grant UID/MAT/04561/2013 to 
Centro de Matemática, Aplicações Fundamentais e Investigação Operacional of Universidade de Lisboa (CMAF-CIO),
 grant UID/EEA/50008/2013 to Instituto de Telecomunicações (IT), and project PTDC/EEI-CTP/4503/2014. Paulo Mateus
 acknowledges IT internal project QBigD.



\end{document}
